\documentclass[12pt]{article}

\usepackage[T1]{fontenc}
\usepackage{lmodern}

\usepackage{setspace}
\usepackage[margin=1.25in]{geometry}

\usepackage[dvipsnames]{xcolor}
\usepackage{pdfpages}
\usepackage{float}
\usepackage{multirow}
\usepackage{caption}
\usepackage{subcaption}
\usepackage{epigraph}

\usepackage{mathtools}
\usepackage{amssymb}
\usepackage{amsthm}
\usepackage{amsfonts}
\usepackage{aliascnt}
\usepackage{accents}
\usepackage{dutchcal}
\usepackage{bbm}

\usepackage{enumitem}
\usepackage{sgame}
\usepackage{tikz}
\usepackage{tikz-cd}
\usetikzlibrary{calc,shapes,arrows}
\usepackage{tcolorbox}
\usepackage{listings}
\usepackage[normalem]{ulem}

\usepackage[round]{natbib}
\newcommand{\E}{\mathbb{E}}

\newtheorem{theorem}{Theorem}[section]

\newaliascnt{proposition}{theorem}
\newtheorem{proposition}[proposition]{Proposition}
\aliascntresetthe{proposition}

\newaliascnt{lemma}{theorem}
\newtheorem{lemma}[lemma]{Lemma}
\aliascntresetthe{lemma}

\newaliascnt{corollary}{theorem}
\newtheorem{corollary}[corollary]{Corollary}
\aliascntresetthe{corollary}

\newaliascnt{claim}{theorem}

\aliascntresetthe{claim}

\theoremstyle{definition}

\newaliascnt{definition}{theorem}
\newtheorem{definition}[definition]{Definition}
\aliascntresetthe{definition}

\newaliascnt{property}{theorem}

\aliascntresetthe{property}

\newaliascnt{example}{theorem}

\aliascntresetthe{example}

\newaliascnt{assumption}{theorem}

\aliascntresetthe{assumption}

\newaliascnt{condition}{theorem}

\aliascntresetthe{condition}

\newaliascnt{question}{theorem}

\aliascntresetthe{question}

\newaliascnt{remark}{theorem}

\aliascntresetthe{remark}

\newaliascnt{remarks}{theorem}

\aliascntresetthe{remarks}

\newaliascnt{aside}{theorem}

\aliascntresetthe{aside}

\newaliascnt{note}{theorem}

\aliascntresetthe{note}

\usepackage{thmtools}
\usepackage{thm-restate}

\usepackage{hyperref}

\hypersetup{
    colorlinks=true,
    linkcolor=OrangeRed,
    filecolor=Thistle,
    urlcolor=Thistle,
    citecolor=Thistle,
}

\usepackage[nameinlink]{cleveref}

\crefname{theorem}{theorem}{theorems}
\Crefname{theorem}{Theorem}{Theorems}
\crefname{proposition}{proposition}{propositions}
\Crefname{proposition}{Proposition}{Propositions}
\crefname{lemma}{lemma}{lemmas}
\Crefname{lemma}{Lemma}{Lemmas}
\crefname{corollary}{corollary}{corollaries}
\Crefname{corollary}{Corollary}{Corollaries}
\crefname{claim}{claim}{claims}
\Crefname{claim}{Claim}{Claims}
\crefname{definition}{definition}{definitions}
\Crefname{definition}{Definition}{Definitions}
\crefname{property}{property}{properties}
\Crefname{property}{Property}{Properties}
\crefname{example}{example}{examples}
\Crefname{example}{Example}{Examples}
\crefname{assumption}{assumption}{assumptions}
\Crefname{assumption}{Assumption}{Assumptions}

\makeatletter
\let\cref@old@isrefconsecutive\cref@isrefconsecutive
\def\cref@isrefconsecutive#1#2{%
  \begingroup
    \def\cref@assumptiontype{assumption}%
    \cref@gettype{#1}{\cref@typea}%
    \ifx\cref@typea\cref@assumptiontype
      \endgroup
      \@cref@refconsecutivefalse
    \else
      \endgroup
      \cref@old@isrefconsecutive{#1}{#2}%
    \fi
}
\makeatother

\crefname{condition}{condition}{conditions}
\Crefname{condition}{Condition}{Conditions}
\crefname{question}{question}{questions}
\Crefname{question}{Question}{Questions}
\crefname{remark}{remark}{remarks}
\Crefname{remark}{Remark}{Remarks}
\crefname{remarks}{remarks}{remarks}
\Crefname{remarks}{Remarks}{Remarks}
\crefname{aside}{aside}{asides}
\Crefname{aside}{Aside}{Asides}
\crefname{note}{note}{notes}
\Crefname{note}{Note}{Notes}
\crefname{appendix}{appendix}{appendices}
\Crefname{appendix}{Appendix}{Appendices}

\newcommand{\secref}[1]{\hyperref[#1]{\S\ref*{#1}}}

\definecolor{backcolour}{rgb}{0.63,0.79,0.95}
\lstdefinestyle{mystyle}{
  backgroundcolor=\color{backcolour},
  basicstyle=\ttfamily\footnotesize,
  breakatwhitespace=false,
  breaklines=true,
  captionpos=b,
  keepspaces=true,
  numbers=left,
  numbersep=5pt,
  showspaces=false,
  showstringspaces=false,
  showtabs=false,
  tabsize=2
}
\begin{document} 
\title{Wherefore BIC = DIC?}
\author{Mark Whitmeyer\thanks{Arizona State University. Email: \href{mailto:mark.whitmeyer@gmail.com}{mark.whitmeyer@gmail.com}. Dedicated to KS. I used ChatGPT as one would an RA (checking proofs, identifying references, scanning for typos, etc.) and solicited feedback from \href{refine.ink}{refine.ink}.}}
\date{\today}
\maketitle

\begin{abstract}
I show that the equivalence in finite social-choice environments of Bayesian incentive compatibility and dominant-strategy incentive compatibility holds universally if and only if every type-dependent difference in the valuation of prizes operates through a single \textit{scalar} type index. Beyond binary domains this universal robustifiability is nongeneric.
\end{abstract}

\section{Introduction}

Bayesian incentive compatibility (BIC) requires truthful reporting to be optimal after averaging over the other agents' types. Dominant-strategy incentive compatibility (DIC) requires truthful reporting to be optimal for every realization of the other agents' reports. I ask when this distinction can always be eliminated without changing the outcomes that matter for welfare and incentives. More precisely, when does every Bayesian incentive-compatible mechanism have a dominant-strategy incentive-compatible counterpart that delivers every type the same interim expected utility and generates the same \textit{ex ante} expected gross social surplus?

I hold fixed a \textit{utility domain}: a finite set of types \(T\), a finite set of prizes \(X\), and a utility function \(u\colon T\times X\to\mathbf R\), but otherwise allow the social-choice environment to vary freely. Namely, the number of agents, the alternatives, the lotteries assigned by each alternative, and the independent full-support prior may all vary. I call the utility domain \textit{universally robustifiable} when the desired BIC-to-DIC replacement exists in every such environment.

I find that the central feature is the ``dimension of the interaction'' between types and prizes. To isolate this interplay, I compare how different types value the same changes in prizes. After subtracting utility terms that depend only on the type or only on the prize, these comparisons form a matrix. I call its rank the \textit{incentive rank} of \(u\). My main result reveals that the domain is universally robustifiable if and only if its incentive rank is at most one. Equivalently, universal robustifiability holds exactly when
\[
u(t,x)=M(t)f(x)+m(t)+g(x)
\]
for some functions \(M,m\colon T\to\mathbf R\) and \(f,g\colon X\to\mathbf R\). That is, every type-dependent difference in the valuation of prizes must operate through a single scalar type index.

When the incentive rank is at most one, I reduce any BIC mechanism to a mechanism with scalar types and gross utilities affine in type, apply the BIC-to-DIC equivalence for that environment \citep[Theorem 2]{gershkov2013equivalence}, and lift the resulting mechanism back to the original type space. In the other direction, when the incentive rank is at least two, I use three types and three prizes to construct a BIC mechanism that no payoff-equivalent DIC mechanism can match. I then reproduce this counterexample in the original utility domain using small (and balanced) perturbations of a common lottery and finally extend it to the full type space under a full-support prior.

\section{Environment and universal robustifiability}

I begin by separating the utility domain, which remains fixed throughout the paper, from the social-choice environments built upon that domain. The utility domain specifies how each possible type evaluates lotteries over prizes. A social-choice environment then specifies the agents, the available alternatives, and the lottery over prizes that each alternative gives to each agent. My universal robustifiability concept (\textit{infra}) requires the BIC-to-DIC replacement to work in \textit{every} such environment.

Let \(T\) be a finite set of types, let \(X\) be a finite set of prizes, and let \(u\colon T\times X\to\mathbf R\). For every finite set \(S\), write \(\Delta(S)\) for the probability distributions on \(S\). Extend \(u(t,\cdot)\) linearly to lotteries and finite signed measures on \(X\).

A finite social-choice environment consists of a finite set of agents \(I\), a finite set of alternatives \(\mathcal A\), and, for every \(i\in I\) and \(a\in\mathcal A\), a lottery \(\lambda_i(a)\in\Delta(X)\).\footnote{To wit, alternative \(a\) provides agent \(i\) with the lottery \(\lambda_i(a)\) over prizes, and the same social alternative may give different lotteries to different agents.} Agent \(i\)'s type \(t_i\in T\) is independently distributed according to a full-support prior \(\pi_i\). Write \(\pi\coloneqq\prod_{i\in I}\pi_i\) and \(\pi_{-i}\coloneqq\prod_{j\neq i}\pi_j\).

A direct mechanism asks each agent to announce an element of \(T\). Write \(\hat t_i\) for agent \(i\)'s report and \(\hat t=(\hat t_j)_{j\in I}\) for the report profile. The (direct) mechanism \((q,p)\) consists of an allocation rule \(q\colon T^I\to\Delta(\mathcal A)\) and payments \(p_i\colon T^I\to\mathbf R\) made by the agents. Write \(q^a\left(\hat{t}\right)\) for the probability of alternative \(a\) after the report profile \(\hat t\).

The interim expected utility of type \(t_i\) from reporting \(\hat t_i\), when the other agents report truthfully, is
\[
U_i(t_i,\hat t_i;q,p)
\coloneqq
\sum_{t_{-i}\in T^{I\setminus\{i\}}}
\pi_{-i}(t_{-i})
\left(
\sum_{a\in\mathcal A}
q^a(\hat t_i,t_{-i})u(t_i,\lambda_i(a))
-p_i(\hat t_i,t_{-i})
\right).
\]
Write \(U_i^{q,p}(t_i)\coloneqq U_i(t_i,t_i;q,p)\) for interim utility when the truth is told.

The mechanism is Bayesian incentive compatible (BIC) if \(U_i^{q,p}(t_i)\geq U_i(t_i,\hat t_i;q,p)\) for every \(i\in I\) and \(t_i,\hat t_i\in T\). It is dominant-strategy incentive compatible (DIC) if
\[\sum_{a\in\mathcal A}q^a(t_i,t_{-i})u(t_i,\lambda_i(a))
-p_i(t_i,t_{-i}) \geq
\sum_{a\in\mathcal A}q^a(\hat t_i,t_{-i})u(t_i,\lambda_i(a))
-p_i(\hat t_i,t_{-i}),
\]
for every \(i\in I\), \(t_i,\hat t_i\in T\), and \(t_{-i}\in T^{I\setminus\{i\}}\).

The mechanism's \textit{ex ante} expected gross social surplus is
\[
W^\pi(q)
\coloneqq
\sum_{t\in T^I}\pi(t)
\sum_{a\in\mathcal A}q^a(t)
\sum_{i\in I}u(t_i,\lambda_i(a)).
\]
Two mechanisms are payoff-equivalent if they give every type of every agent the same interim expected utility and generate the same \textit{ex ante} expected gross social surplus.

\begin{definition}\label{def:universal-robustifiability}
The domain \((T,X,u)\) is \textit{universally robustifiable} if, in every finite social-choice environment and under every independent full-support prior, every BIC mechanism has a payoff-equivalent DIC mechanism.
\end{definition}

\section{Incentive rank}

The agents' incentives depend on how the attractiveness of different prizes changes with their types. Naturally, a utility term that depends only on the type is the same under every report and is, thus, irrelevant for incentives. In contrast, a utility term that depends only on the prize may affect which alternatives are attractive, but it affects every type in the same way and does not contribute to variation in preferences across types. 

To compare how different types value prizes, I consider their utility gains from the same change in prize. To ease this comparison, I choose one type \(t_0\in T\) and one prize \(x_0\in X\) as baselines for comparison; I call them the baseline type and baseline prize. For \(t\in T\) and \(x\in X\), define
\[
\Delta_u^{t_0,x_0}(t,x)
\coloneqq
u(t,x)-u(t,x_0)-u(t_0,x)+u(t_0,x_0).
\]
Accordingly, \(\Delta_u^{t_0,x_0}(t,x)\) is type \(t\)'s gain in utility from replacing \(x_0\) by \(x\), minus the baseline type's utility gain from the same replacement.

For fixed baselines \(t_\circ\in T\) and \(x_\circ\in X\), I call \(D^{t_\circ,x_\circ} \coloneqq (
\Delta_u^{t_\circ,x_\circ}(t,x)
)_{t\in T,x\in X}\) the \textit{utility matrix associated with \((t_\circ, x_\circ)\)} (or just the \textit{utility matrix} when unambiguous) and define the \textit{incentive rank} of \(u\) by \(r(u)
\coloneqq
\operatorname{rank} D^{t_\circ,x_\circ}\).\footnote{Here, \(\operatorname{rank}\) denotes the usual matrix rank; \textit{viz.}, the number of linearly independent rows or columns of the utility matrix.} Economically, the incentive rank is the smallest number of type-specific indices needed to describe how utility differences across prizes vary with type. Rank zero means that type affects utility levels but not utility differences across prizes. Rank one means that every prize has a common index \(f(x)\), and types differ in their tradeoffs among prizes only through the coefficient \(M(t)\) attached to that index. Namely, once \(M(t)\) is known, the type label contains no further information relevant for reporting incentives. Rank at least two means that no single prize index and type-specific coefficient can capture all type-dependent tradeoffs among prizes.

For \(t_0,t_1,t_2\in T\) and \(x_0,x_1,x_2\in X\), define
\[
\Omega_u(t_0,t_1,t_2;x_0,x_1,x_2)
\coloneqq
\det
\begin{pmatrix}
\Delta_u^{t_0,x_0}(t_1,x_1)&\Delta_u^{t_0,x_0}(t_1,x_2)\\
\Delta_u^{t_0,x_0}(t_2,x_1)&\Delta_u^{t_0,x_0}(t_2,x_2)
\end{pmatrix}.
\]
I call \((t_0,t_1,t_2;x_0,x_1,x_2)\) a utility triangle and \(\Omega_u(t_0,t_1,t_2;x_0,x_1,x_2)\) its determinant.

\begin{lemma}\label{lem:incentive-rank-characterization}
The incentive rank does not depend on the baselines \(t_\circ\) and \(x_\circ\). Moreover, the following are equivalent:
\begin{enumerate}[nosep,noitemsep]
\item \(r(u)\leq1\).
\item There are functions \(M,m\colon T\to\mathbf R\) and \(f,g\colon X\to\mathbf R\) such that
\[
u(t,x)=M(t)f(x)+m(t)+g(x), \quad \forall \ t \in T, \ \forall \ x\in X.
\tag{1}\label{eq:rank-one-utility-factorization}
\]
\item \label{it:ir3} \(\Omega_u(t_0,t_1,t_2;x_0,x_1,x_2)=0\) for every \(t_0,t_1,t_2\in T\) and \(x_0,x_1,x_2\in X\).
\end{enumerate}
\end{lemma}
\Cref{lem:incentive-rank-characterization} reveals that one scalar \(M(t)\) suffices exactly when every utility triangle has zero determinant. In that case, after removing a type-only term and a type-independent prize term, every type-dependent difference in the valuation of prizes can be written as \(M(t)f(x)\).

\begin{proof}[Proof of \Cref{lem:incentive-rank-characterization}]
Fix two pairs of baselines \((t_0,x_0)\) and \((\hat t_0,\hat x_0)\). Directly,
\[\Delta_u^{\hat t_0,\hat x_0}(t,x) =\Delta_u^{t_0,x_0}(t,x) -\Delta_u^{t_0,x_0}(\hat t_0,x) -\Delta_u^{t_0,x_0}(t,\hat x_0) +\Delta_u^{t_0,x_0}(\hat t_0,\hat x_0).
\]
i.e., I obtain \(D^{\hat t_0,\hat x_0}\) from \(D^{t_0,x_0}\) by subtracting the row indexed by \(\hat t_0\) from every row and then subtracting the resulting column indexed by \(\hat x_0\) from every column. These operations cannot increase rank. Interchanging the two pairs of baselines produces the reverse rank inequality. Hence, the two utility matrices have the same rank.

Suppose that \(r(u)\leq1\), and fix baselines \(t_0\in T\) and \(x_0\in X\). Write \(D(t,x)\coloneqq\Delta_u^{t_0,x_0}(t,x)\). If \(D\) is identically zero, let \(M\) and \(f\) be the zero functions. Otherwise, choose \(\hat t\in T\) and \(\hat x\in X\) such that \(D(\hat t,\hat x)\neq0\). Because \(D\) has rank one, every \(2\times2\) minor vanishes, and so
\(D(t,x)D(\hat t,\hat x)
=
D(t,\hat x)D(\hat t,x)\) for every \(t\in T\) and \(x\in X\). 

Define \(M(t)\coloneqq D(t,\hat x)\) and \(f(x)\coloneqq\frac{D(\hat t,x)}{D(\hat t,\hat x)}\) so that \(D(t,x)=M(t)f(x)\). Setting \(m(t)\coloneqq u(t,x_0)\) and \(g(x)\coloneqq u(t_0,x)-u(t_0,x_0)\) produces 
\[M(t)f(x)+m(t)+g(x) =D(t,x)+u(t,x_0)+u(t_0,x)-u(t_0,x_0) =u(t,x),
\]
and so \(r(u)\leq1\) implies \eqref{eq:rank-one-utility-factorization}.

Conversely, suppose that \eqref{eq:rank-one-utility-factorization} holds. For arbitrary baselines \(t_0\in T\) and \(x_0\in X\),
\[
\Delta_u^{t_0,x_0}(t,x)
=
(M(t)-M(t_0))(f(x)-f(x_0)).
\]
The utility matrix is an outer product and has rank at most one.\footnote{To elaborate, it is the outer product of the column vector \((M(t)-M(t_0))_{t\in T}\) and the row vector \((f(x)-f(x_0))_{x\in X}\); and so every row is a scalar multiple of \((f(x)-f(x_0))_{x\in X}\).} Hence, \eqref{eq:rank-one-utility-factorization} implies \(r(u)\leq1\).

If \(r(u)\leq1\), every \(2\times2\) minor of every utility matrix vanishes, so every utility triangle has zero determinant.

Finally, suppose that every utility triangle has zero determinant. Fix the baselines \(t_\circ\) and \(x_\circ\) used to define \(r(u)\). Every \(2\times2\) minor of \(D^{t_\circ,x_\circ}\) is obtained by choosing two rows \(t_1,t_2\in T\) and two columns \(x_1,x_2\in X\). Its determinant is exactly \(\Omega_u(t_\circ,t_1,t_2;x_\circ,x_1,x_2)\), which is zero by hypothesis. Hence, every \(2\times2\) minor vanishes, so \(r(u)\leq1\).\footnote{A matrix has rank at most one exactly when all its \(2\times2\) minors vanish.}
\end{proof}


\begin{theorem}\label{thm:maximal-domain}
The following are equivalent:
\begin{enumerate}[noitemsep,nosep]
\item The domain \((T,X,u)\) is universally robustifiable.
\item \(r(u)\leq1\).
\item The factorization in \eqref{eq:rank-one-utility-factorization} holds.
\item Every utility triangle has zero determinant.
\end{enumerate}
If \(r(u)\geq2\), there is a two-agent environment with three lottery-valued alternatives and an independent full-support prior on \(T^2\) in which a BIC mechanism with a deterministic allocation rule has no payoff-equivalent DIC mechanism, even when DIC allocation rules are allowed to randomize. Moreover, there is a lottery \(\lambda^\circ\in\Delta(X)\) such that, for each alternative \(k\), the lotteries \(\lambda_k^+,\lambda_k^-\in\Delta(X)\) assigned to the two agents satisfy \(\lambda_k^++\lambda_k^-=2\lambda^\circ\) as finite measures on \(X\).
\end{theorem}
When \(r(u)\leq1\), the factorization in \eqref{eq:rank-one-utility-factorization} reduces every social-choice environment to one with scalar types and gross utilities affine in type. I average over any original type labels that share the same scalar value, apply \citet{gershkov2013equivalence}'s BIC-to-DIC equivalence in the scalar environment, then lift the resulting DIC mechanism back to the original types.

When \(r(u)\geq2\), some three types and three prizes form a utility triangle with nonzero determinant. This allows me to carefully calibrate the probabilities of those three prizes to reproduce the two patterns of type-dependent utility differences I use in the counterexample below. First, in \Cref{lem:balanced-obstruction}, I construct that three-type counterexample. Next, I reproduce it in \Cref{lem:balanced-embedding} by using lotteries over the original prizes. Finally, in \Cref{lem:full-support-extension}, I assign positive probability to all remaining types while preserving the fact that no payoff-equivalent DIC mechanism exists.

\begin{corollary}\label{cor:binary-domains}
If \(|T|\leq2\) or \(|X|\leq2\), then \((T,X,u)\) is universally robustifiable for every utility function \(u\colon T\times X\to\mathbf R\). Conversely, if \(|T|\geq3\) and \(|X|\geq3\), there is a utility function on \(T\times X\) that is not universally robustifiable.
\end{corollary}
And so, failure first becomes possible with three types and three prizes.
\begin{proof}
Fix a baseline type and prize. The associated utility matrix has a zero row and a zero column, so \(r(u)\leq\min\{|T|-1,|X|-1\}\). Hence, if \(|T|\leq2\) or \(|X|\leq2\), then \(r(u)\leq1\), and \Cref{thm:maximal-domain} implies universal robustifiability.

Now suppose that \(|T|\geq3\) and \(|X|\geq3\). Choose distinct \(t_0,t_1,t_2\in T\) and \(x_0,x_1,x_2\in X\), define
\[
(u(t_i,x_j))_{i,j=0}^2=
\begin{pmatrix}
0&0&0\\
0&1&0\\
0&0&1
\end{pmatrix},
\]
and set \(u(t,x)=0\) for all remaining pairs. The utility triangle \((t_0,t_1,t_2;x_0,x_1,x_2)\) has determinant one. By \Cref{thm:maximal-domain}, the resulting domain is not universally robustifiable.
\end{proof}

\begin{corollary}\label{cor:universal-robustifiability-nongeneric}
Suppose that \(|T|\geq3\) and \(|X|\geq3\). Viewed as a subset of \(\mathbf R^{T\times X}\), the universally robustifiable utility functions form a closed algebraic set of codimension \((|T|-2)(|X|-2)\). In particular, this set has empty interior and Lebesgue measure zero.
\end{corollary}
And so, beyond binary domains, universal robustifiability is nongeneric.
\begin{proof}
Fix baselines \(t_0\in T\) and \(x_0\in X\), and let \(\Phi(u)\) be the matrix obtained by deleting the zero baseline row and zero baseline column from \(D^{t_0,x_0}\). The linear map \(\Phi\colon\mathbf R^{T\times X}\to\mathbf R^{(|T|-1)\times(|X|-1)}\) is onto: any matrix in the codomain arises by setting the baseline row and column of \(u\) equal to zero and using the prescribed matrix for the remaining entries.

By \Cref{thm:maximal-domain}, \(u\) is universally robustifiable exactly when \(\Phi(u)\) has rank at most one, or equivalently, when every \(2\times2\) minor of \(\Phi(u)\) vanishes. Hence, the universally robustifiable utility functions form the inverse image under \(\Phi\) of the rank-at-most-one determinantal variety.

Set \(m\coloneqq|T|-1\) and \(n\coloneqq|X|-1\). The rank-at-most-one matrices in \(\mathbf R^{m\times n}\) have dimension \(m+n-1\), thus, codimension \(mn-(m+n-1)=(m-1)(n-1)\). Because \(\Phi\) is onto, taking its inverse image preserves codimension. This codimension is positive, so the universally robustifiable utility functions have empty interior and Lebesgue measure zero.
\end{proof}

\section{Sufficiency}

Suppose that \(r(u)\leq1\). The factorization in \eqref{eq:rank-one-utility-factorization} says that all type-dependent reporting tradeoffs operate through the scalar \(M(t)\). Of course, types with the same value of \(M(t)\) may still have different labels and different type-only utility terms, but those differences do not affect their preferences among reports.

There are four steps in the proof of sufficiency. First, I rewrite utility so that the coefficient multiplying \(M(t)\) is nonnegative for every alternative. Second, I collapse the original type space to the scalar space \(M(T)\). Third, because several original types may have the same scalar value, I average the original mechanism over their labels to obtain a well-defined BIC mechanism for the scalar type space. Fourth, I apply the scalar-type BIC-to-DIC result \citep{gershkov2013equivalence} and lift its conclusion back to the original type space.

\begin{proposition}\label{prop:rank-one-robustification}
Suppose that \(r(u)\leq1\). Then \((T,X,u)\) is universally robustifiable.
\end{proposition}

\begin{proof}[Proof of \Cref{prop:rank-one-robustification}] 
Fix a finite social-choice environment, an independent full-support prior, and a BIC mechanism \((q,p)\). By \Cref{lem:incentive-rank-characterization}, choose \(M,m,f,g\) satisfying \eqref{eq:rank-one-utility-factorization}. For \(\lambda\in\Delta(X)\), define \(f(\lambda)\coloneqq\sum_{x\in X}\lambda(x)f(x)\) and \(g(\lambda)\coloneqq\sum_{x\in X}\lambda(x)g(x)\). Then \(u(t,\lambda)=\sum_{x\in X}\lambda(x)u(t,x)=M(t)f(\lambda)+m(t)+g(\lambda)\) for every \(t\in T\) and \(\lambda\in\Delta(X)\).

For every agent \(i\in I\), define \(\underline f_i
\coloneqq
\min_{a\in\mathcal A}f(\lambda_i(a))\), \(a_i^a
\coloneqq
f(\lambda_i(a))-\underline f_i\), \(c_i^a\coloneqq g(\lambda_i(a))\), and \(\mu_i(t)\coloneqq m(t)+\underline f_iM(t)\). Accordingly, \(a_i^a\geq0\), and
\[
u(t,\lambda_i(a))
=
a_i^aM(t)+\mu_i(t)+c_i^a.
\tag{2}\label{eq:rank-one-linear-utility}
\]
Let \(\Theta\coloneqq M(T)\), let \(\overline\pi_i\) be the distribution of \(\theta_i=M(t_i)\) induced by \(\pi_i\), and write \(\overline\pi\coloneqq\prod_{i\in I}\overline\pi_i\). For each \(\theta\in\Theta\), let \(\pi_i(\cdot\mid\theta)\) denote the conditional distribution of \(t_i\) given \(M(t_i)=\theta\).

For type \(\theta\in\Theta\) and a report \(t_i\in T\), define the part of interim utility that omits the report-independent term \(\mu_i\) by
\[
V_i(\theta,t_i)
\coloneqq
\sum_{t_{-i}\in T^{I\setminus\{i\}}}
\pi_{-i}(t_{-i})
\left(
\sum_{a\in\mathcal A}
q^a(t_i,t_{-i})
(a_i^a\theta+c_i^a)
-p_i(t_i,t_{-i})
\right).
\]
If \(t_i\) and \(t_i'\) satisfy \(M(t_i)=M(t_i')=\theta\), BIC applied in both directions produces \(V_i(\theta,t_i)=V_i(\theta,t_i')\). Denote this common value by \(V_i(\theta)\). More generally, BIC implies \(V_i(\theta)\geq V_i(\theta,\hat t_i)\) for every \(\hat t_i\in T\). \eqref{eq:rank-one-linear-utility} also yields
\[
U_i^{q,p}(t_i)
=
\mu_i(t_i)+V_i(M(t_i)).
\tag{3}\label{eq:rank-one-original-interim-utility}
\]

Construct a direct mechanism \((\overline q,\overline p)\) on \(\Theta\) by independently drawing, conditional on every reported scalar type \(\hat\theta_i\), a label \(\hat t_i\) according to \(\pi_i(\cdot\mid\hat\theta_i)\), and then applying \((q,p)\). Equivalently,
\[
\overline q^a(\hat\theta)
\coloneqq
\sum_{\substack{\hat t\in T^I\\M(\hat t_i)=\hat\theta_i\ \forall i}}
\prod_{i\in I}\pi_i(\hat t_i\mid\hat\theta_i)
q^a\left(\hat{t}\right) \quad \text{and} \quad
\overline p_i(\hat\theta)
\coloneqq
\sum_{\substack{\hat t\in T^I\\M(\hat t_j)=\hat\theta_j\ \forall j}}
\prod_{j\in I}\pi_j(\hat t_j\mid\hat\theta_j)
p_i\left(\hat{t}\right).
\]
Truthfully reporting \(\theta\) yields the reduced interim utility \(V_i(\theta)\), while reporting \(\hat\theta\) yields an average of \(V_i(\theta,\hat t_i)\) over (original-type) labels satisfying \(M(\hat t_i)=\hat\theta\). Hence, \((\overline q,\overline p)\) is BIC in the scalar-type environment with gross utilities \(a_i^a\theta_i+c_i^a\).

Because the original types are independent, drawing each label conditional on its scalar type reproduces the original joint distribution of labels. Consequently,
\[\sum_{\theta\in\Theta^I}
\overline\pi(\theta)
\sum_{a\in\mathcal A}
\overline q^a(\theta)
\sum_{i\in I}(a_i^a\theta_i+c_i^a) =
\sum_{t\in T^I}
\pi(t)
\sum_{a\in\mathcal A}
q^a(t)
\sum_{i\in I}(a_i^aM(t_i)+c_i^a).
\tag{4}\label{eq:rank-one-reduced-surplus}
\]

The reduced environment has independent one-dimensional types with finite supports and coefficients \(a_i^a\geq0\). Apply \citet[Theorem 2]{gershkov2013equivalence} to the mechanism \((\overline q,-\overline p)\), where the second component is the vector of transfers received by the agents. The theorem provides a DIC mechanism \((q^\star,\tau^\star)\) that preserves every scalar type's interim expected utility and \textit{ex ante} expected gross social surplus. Set \(p_i^\star\coloneqq-\tau_i^\star\) for every agent \(i\).

Lift this mechanism to the original type space by setting \(\tilde q^a(t) \coloneqq q^{\star,a}((M(t_j))_{j\in I})\) and \(\tilde p_i(t) \coloneqq p_i^\star((M(t_j))_{j\in I})\). For a true type \(t_i\), the term \(\mu_i(t_i)\) in \eqref{eq:rank-one-linear-utility} does not depend on the report or alternative. Consequently, DIC of \((q^\star,p^\star)\) implies DIC of \((\tilde q,\tilde p)\).

The lifted mechanism gives type \(t_i\) interim utility \(\mu_i(t_i)+V_i(M(t_i))\), which equals the interim utility of \(t_i\) under \((q,p)\) by \eqref{eq:rank-one-original-interim-utility}. Expected gross social surplus under \((q,p)\) equals
\[
\sum_{i\in I}\E_{\pi_i}[\mu_i(t_i)]
+
\sum_{t\in T^I}\pi(t)
\sum_{a\in\mathcal A}q^a(t)
\sum_{i\in I}(a_i^aM(t_i)+c_i^a),
\]
while expected gross social surplus under \((\tilde q,\tilde p)\) equals \(\sum_{i\in I}\E_{\pi_i}[\mu_i(t_i)]\) plus the expected gross social surplus generated by \(q^\star\) in the reduced environment. The DIC mechanism supplied by \citet[Theorem 2]{gershkov2013equivalence} has an allocation rule \(q^\star\) that generates the same expected gross social surplus as \(\overline q\) in that environment. Hence, \eqref{eq:rank-one-reduced-surplus} implies that \((q,p)\) and \((\tilde q,\tilde p)\) generate the same expected gross social surplus, so they are payoff-equivalent.\end{proof}

\section{Necessity}

\subsection{A three-type obstruction}\label{sec:3to}
I now turn to necessity. I first construct a three-type example in which a BIC mechanism has no payoff-equivalent DIC mechanism. Consider two agents with common set of types \(\mathcal T\coloneqq\{0,1,2\}\) and set of alternatives \(\mathcal K\coloneqq\{0,1,2\}\). Types are independently and uniformly distributed. Let
\[
H\coloneqq
\begin{pmatrix}
0&0&0\\
0&1&2\\
0&-1&1
\end{pmatrix}.
\]
If agent \(1\) has type \(r\in\mathcal T\) and alternative \(k\in\mathcal K\) is selected, her gross utility is \(H_{rk}\). If agent \(2\) has type \(s\in\mathcal T\), his gross utility is \(-H_{sk}\). Let \(\mathcal E_H\) denote the two-agent environment with common type set \(\mathcal T\), alternative set \(\mathcal K\), independent uniform priors, and gross utilities \(H_{rk}\) and \(-H_{sk}\) for agents \(1\) and \(2\), respectively.

A direct mechanism consists of an allocation rule \(q\) and payments \(p_i(r,s)\) made by the agents. For \((r,s)\in\mathcal T^2\), write \(q_{rs}\in\Delta(\mathcal K)\) for the lottery over alternatives and \(q_{rs}^k\) for the probability of alternative \(k\). Write
\[v^1_{rs} \coloneqq\sum_{k\in\mathcal K}q_{rs}^kH_{rk}-p_1(r,s) \qquad \text{and} \qquad v^2_{rs} \coloneqq-\sum_{k\in\mathcal K}q_{rs}^kH_{sk}-p_2(r,s)\]
for the truthful \textit{ex post} utilities. The truthful interim utilities are
\[
U_1(r)\coloneqq\frac{1}{3}\sum_{s\in\mathcal T}v^1_{rs}
\qquad \text{and} \qquad
U_2(s)\coloneqq\frac{1}{3}\sum_{r\in\mathcal T}v^2_{rs},
\]
and expected gross social surplus is
\(W(q)
\coloneqq
\frac{1}{9}
\sum_{r,s\in\mathcal T}
\sum_{k\in\mathcal K}
q_{rs}^k(H_{rk}-H_{sk})\).

\begin{lemma}\label{lem:balanced-obstruction}
There is a deterministic BIC mechanism with interim utility vectors
\[
U_1^\star=\left(0,\frac{1}{3},0\right)
\qquad \text{and} \qquad
U_2^\star=(0,0,0),
\]
and expected gross social surplus \(W(q^\star)=\frac{2}{3}\). Every DIC mechanism with the same interim utility vectors satisfies \(W(q)\leq\frac{5}{9}\).
\end{lemma}
That is, the BIC mechanism has no payoff-equivalent DIC mechanism. In the proof, I first display a deterministic BIC mechanism and compute its interim utilities and surplus. I then derive two inequalities that every DIC mechanism with the same interim utilities must satisfy and use them to show that its surplus is at most \(5/9\).

\begin{proof}[Proof of \Cref{lem:balanced-obstruction}]
Let \(q^\star\) choose \(k^\star(r,s)\) with probability one, where
\[
\begin{array}{c|ccc}
k^\star(r,s)&s=0&s=1&s=2\\
\hline
r=0&0&0&1\\
r=1&2&0&1\\
r=2&2&0&0
\end{array}.
\]
Charge agent \(1\) the payments \(p_1^\star(1,2)=2\) and \(p_1^\star(2,2)=1\), and charge agent \(2\) the payment \(p_2^\star(2,2)=2\). All other payments are zero.

Let \(\mathcal U_1(r,\hat r)\) denote the interim utility of agent \(1\) with true type \(r\) from reporting \(\hat r\), and define \(\mathcal U_2(s,\hat s)\) analogously. Directly,
\[
(\mathcal U_1(r,\hat r))_{r,\hat r\in\mathcal T}
=
\frac{1}{3}
\begin{pmatrix}
0&-2&-1\\
1&1&1\\
-1&-2&0
\end{pmatrix}
\quad \text{and} \quad
(\mathcal U_2(s,\hat s))_{s,\hat s\in\mathcal T}
=
\frac{1}{3}
\begin{pmatrix}
0&0&-2\\
-4&0&-4\\
-2&0&0
\end{pmatrix}.
\]
Each diagonal entry is (weakly) maximal in its row. The mechanism is, therefore, BIC, with \(U_1^\star=(0,\frac{1}{3},0)\) and \(U_2^\star=(0,0,0)\).

The gross social surplus generated by \(q^\star\) at each type profile is
\[
(H_{r,k^\star(r,s)}-H_{s,k^\star(r,s)})_{r,s\in\mathcal T}
=
\begin{pmatrix}
0&0&1\\
2&0&2\\
1&0&0
\end{pmatrix}.
\]
Consequently, \(W(q^\star)=\frac{6}{9}=\frac{2}{3}\).

Now consider any DIC mechanism \((q,p)\) with the same interim utility vectors. Its truthful \textit{ex post} utilities satisfy
\[
\begin{aligned}
&\sum_{s\in\mathcal T}v^1_{0s}=0,
\qquad
\sum_{s\in\mathcal T}v^1_{1s}=1,
\qquad
\sum_{s\in\mathcal T}v^1_{2s}=0,\\
&\sum_{r\in\mathcal T}v^2_{r0}=0,
\qquad
\sum_{r\in\mathcal T}v^2_{r1}=0,
\qquad
\sum_{r\in\mathcal T}v^2_{r2}=0.
\end{aligned}
\tag{5}\label{eq:balanced-target-interim-sums}
\]

For agent \(1\), DIC is equivalent to
\[
v^1_{rs}-v^1_{\hat r s}
\geq
\sum_{k\in\mathcal K}
q_{\hat r s}^k(H_{rk}-H_{\hat r k}) \quad \text{for every } r,\hat{r}, s\in\mathcal T
\tag{6}\label{eq:balanced-agent-one-dic}
\]
Applying \eqref{eq:balanced-agent-one-dic} at \((r,\hat r,s) \in \{(1,0,1),(1,0,2),(1,2,0),(2,0,0)\}\) produces
\[
\begin{aligned}
v^1_{11}-v^1_{01}&\geq q_{01}^1+2q_{01}^2, \qquad
v^1_{12}-v^1_{02} \geq q_{02}^1+2q_{02}^2,\\
v^1_{10}-v^1_{20}&\geq2q_{20}^1+q_{20}^2,\qquad
v^1_{20}-v^1_{00} \geq-q_{00}^1+q_{00}^2.
\end{aligned}
\]
Adding the four inequalities obtained from \eqref{eq:balanced-agent-one-dic} at \((1,0,1)\), \((1,0,2)\), \((1,2,0)\), and \((2,0,0)\) and then applying \eqref{eq:balanced-target-interim-sums} to the resulting left-hand side yields
\[A(q) \coloneqq
q_{01}^1+2q_{01}^2
+q_{02}^1+2q_{02}^2
+2q_{20}^1+q_{20}^2 -q_{00}^1+q_{00}^2 \leq 1.
\tag{7}\label{eq:balanced-agent-one-certificate}
\]

For agent \(2\), DIC is equivalent to
\[
v^2_{rs}-v^2_{r\hat s}
\geq
\sum_{k\in\mathcal K}
q_{r\hat s}^k(H_{\hat s k}-H_{sk}) \quad \text{for every } r,s,\hat s\in\mathcal T.
\tag{8}\label{eq:balanced-agent-two-dic}
\]
Applying \eqref{eq:balanced-agent-two-dic} at \((r,s,\hat s)
\in
\{(0,0,1),(0,2,0),(1,2,1),(2,2,1)\}\) produces \[
\begin{aligned}
v^2_{00}-v^2_{01}&\geq q_{01}^1+2q_{01}^2, \qquad
v^2_{02}-v^2_{00} \geq q_{00}^1-q_{00}^2,\\
v^2_{12}-v^2_{11}&\geq2q_{11}^1+q_{11}^2,\qquad
v^2_{22}-v^2_{21} \geq2q_{21}^1+q_{21}^2.
\end{aligned}
\]
Adding the four inequalities obtained from \eqref{eq:balanced-agent-two-dic} at \((0,0,1)\), \((0,2,0)\), \((1,2,1)\), and \((2,2,1)\) and then applying \eqref{eq:balanced-target-interim-sums} to the resulting left-hand side yields
\[B(q) \coloneqq
q_{01}^1+2q_{01}^2
+q_{00}^1-q_{00}^2
+2q_{11}^1+q_{11}^2
+2q_{21}^1+q_{21}^2 \leq0.
\tag{9}\label{eq:balanced-agent-two-certificate}
\]

Set
\[\Gamma(q)\coloneqq -q_{01}^1-2q_{01}^2
+q_{02}^1-q_{02}^2
-q_{20}^1+q_{20}^2 -2q_{21}^1-q_{21}^2.
\]
Substituting the entries of \(H\) into the definition of \(W(q)\) and collecting terms yields
\[
9W(q)
=
q_{10}^1+2q_{10}^2
+2q_{12}^1+q_{12}^2
+\Gamma(q).
\tag{10}\label{eq:balanced-surplus-expansion}
\]
Because \(q_{rs}\in\Delta(\mathcal K)\),
\(q_{10}^1+2q_{10}^2\leq2\) and \(2q_{12}^1+q_{12}^2\leq2\). Moreover,
\[A(q)+B(q)-\Gamma(q)= 3q_{01}^1+6q_{01}^2
+3q_{02}^2
+2q_{11}^1+q_{11}^2 +3q_{20}^1
+4q_{21}^1+2q_{21}^2 \geq 0.
\]
The expansion of \(A(q)+B(q)-\Gamma(q)\) has only nonnegative terms, so \(\Gamma(q)\leq A(q)+B(q)\). Substituting \(q_{10}^1+2q_{10}^2\leq2\), \(2q_{12}^1+q_{12}^2\leq2\), and \(\Gamma(q)\leq A(q)+B(q)\) into \eqref{eq:balanced-surplus-expansion} and then applying \eqref{eq:balanced-agent-one-certificate} and \eqref{eq:balanced-agent-two-certificate} produces \(9W(q)\leq2+2+A(q)+B(q)\leq5\). Consequently, \(W(q)\leq\frac{5}{9}<\frac{2}{3}=W(q^\star)\), so no DIC mechanism with the same interim utility vectors is payoff-equivalent to \((q^\star,p^\star)\).
\end{proof}

\subsection{Embedding the obstruction}


Retain the matrix \(H\), the type set \(\mathcal T\), and the alternative set \(\mathcal K\) used in \secref{sec:3to} to define \(\mathcal E_H\). Fix types \(t_0,t_1,t_2\in T\) and prizes \(x_0,x_1,x_2\in X\). Define
\[
D^\circ
\coloneqq
\begin{pmatrix}
\Delta_u^{t_0,x_0}(t_1,x_1)&\Delta_u^{t_0,x_0}(t_1,x_2)\\
\Delta_u^{t_0,x_0}(t_2,x_1)&\Delta_u^{t_0,x_0}(t_2,x_2)
\end{pmatrix}.
\]
Suppose that \(\det D^\circ\neq0\). Let
\[
C\coloneqq(H_{rk})_{r,k\in\{1,2\}}
=
\begin{pmatrix}
1&2\\
-1&1
\end{pmatrix}
\qquad\text{and}\qquad
L\coloneqq(D^\circ)^{-1}C.
\]
Define the signed measures \(z_0,z_1,z_2\) by \(z_0\coloneqq0\) and
\(z_k\coloneqq\sum_{j=1}^2L_{jk}(\delta_{x_j}-\delta_{x_0})\) for
\(k\in\{1,2\}\).\footnote{Here, \(\delta_x\) assigns probability one to \(x\). Each \(z_k\) has total mass zero.} Set \(\beta_k\coloneqq u(t_0,z_k)\) for every \(k\in\mathcal K\).

Choose a lottery \(\lambda^\circ\in\Delta(X)\) that assigns strictly positive probability to \(x_0,x_1,x_2\). For sufficiently small \(\varepsilon>0\), define
\(\lambda_k^+\coloneqq\lambda^\circ+\varepsilon z_k\) and \(\lambda_k^-\coloneqq\lambda^\circ-\varepsilon z_k\), which belong to \(\Delta(X)\) for every \(k\in\mathcal K\). For \(r\in\mathcal T\), set \(a_r\coloneqq u(t_r,\lambda^\circ)\).

Consider the two-agent environment in which both agents have the set of types \(\{t_0,t_1,t_2\}\), types are independently and uniformly distributed, and alternative \(k\) assigns \(\lambda_k^+\) to agent \(1\) and \(\lambda_k^-\) to agent \(2\). Identify each \(r\in\mathcal T\) with \(t_r\), and write \(\mathcal E_\lambda\) for this environment.

\begin{lemma}\label{lem:balanced-embedding}
There is a bijection between direct mechanisms in \(\mathcal E_H\) and \(\mathcal E_\lambda\) that preserves allocation rules, BIC, DIC, and payoff equivalence. Moreover, \(\lambda_k^++\lambda_k^-=2\lambda^\circ\) as finite measures on \(X\) for every \(k\in\mathcal K\).
\end{lemma}

\begin{proof}[Proof of \Cref{lem:balanced-embedding}]
For \(r,k\in\{1,2\}\), by the definitions of \(D^\circ\), \(L\), and \(z_k\),
\[u(t_r,z_k)-u(t_0,z_k) =\sum_{j=1}^2L_{jk}\Delta_u^{t_0,x_0}(t_r,x_j) =(D^\circ L)_{rk} =C_{rk} =H_{rk}.\]
When \(r=0\), both \(u(t_r,z_k)-u(t_0,z_k)\) and \(H_{rk}\) equal zero. When \(k=0\), \(z_0=0\), so \(\beta_0=0\) and \(H_{r0}=0\). Hence,
\[
u(t_r,z_k)=\beta_k+H_{rk}
\quad\text{for every }r\in\mathcal T\text{ and }k\in\mathcal K.
\tag{11}\label{eq:balanced-direction-utilities}
\]

By construction, \(\lambda_k^++\lambda_k^-=2\lambda^\circ\) as finite measures on \(X\). Linearity of expected utility and \eqref{eq:balanced-direction-utilities} imply
\[u(t_r,\lambda_k^+) =a_r+\varepsilon\beta_k+\varepsilon H_{rk} \qquad \text{and} \qquad u(t_r,\lambda_k^-) =a_r-\varepsilon\beta_k-\varepsilon H_{rk}.
\tag{12}\label{eq:balanced-gross-utility-embedding}
\]

Fix a direct mechanism \((q,p)\) in \(\mathcal E_H\). Define payments in \(\mathcal E_\lambda\) by
\[\tilde p_1(\hat r,\hat s) \coloneqq
\varepsilon p_1(\hat r,\hat s)
+\varepsilon\sum_{k\in\mathcal K}\beta_kq_{\hat r\hat s}^k,\qquad \text{and} \qquad
\tilde p_2(\hat r,\hat s) \coloneqq
\varepsilon p_2(\hat r,\hat s)
-\varepsilon\sum_{k\in\mathcal K}\beta_kq_{\hat r\hat s}^k.
\tag{13}\label{eq:balanced-payment-embedding}
\]
For each fixed allocation rule, \eqref{eq:balanced-payment-embedding} is invertible and, therefore, defines a bijection between mechanisms.

For true types \(t_r,t_s\) and reports \((\hat r,\hat s)\), equations \eqref{eq:balanced-gross-utility-embedding} and \eqref{eq:balanced-payment-embedding} yield
\[
\begin{aligned}
\sum_{k\in\mathcal K}q_{\hat r\hat s}^k
u(t_r,\lambda_k^+)
-\tilde p_1(\hat r,\hat s) &=
a_r+\varepsilon
\left(
\sum_{k\in\mathcal K}q_{\hat r\hat s}^kH_{rk}
-p_1(\hat r,\hat s)
\right) \quad \text{and} \\
\sum_{k\in\mathcal K}q_{\hat r\hat s}^k
u(t_s,\lambda_k^-)
-\tilde p_2(\hat r,\hat s) &= 
a_s+\varepsilon
\left(
-\sum_{k\in\mathcal K}q_{\hat r\hat s}^kH_{sk}
-p_2(\hat r,\hat s)
\right).
\end{aligned}
\tag{14}\label{eq:balanced-report-utility-embedding}
\]
The terms \(a_r\) and \(a_s\) depend only on the true types, and \(\varepsilon>0\). Thus, every pointwise reporting inequality in \(\mathcal E_H\) is equivalent to the corresponding inequality in \(\mathcal E_\lambda\). The bijection preserves DIC. Taking expectations over the opponent's type in \eqref{eq:balanced-report-utility-embedding} shows that it also preserves BIC.

If \(U_i\) and \(\tilde U_i\) denote truthful interim utilities in \(\mathcal E_H\) and \(\mathcal E_\lambda\), respectively, then
\[
\tilde U_1(t_r)=a_r+\varepsilon U_1(r)
\qquad \text{and} \qquad
\tilde U_2(t_s)=a_s+\varepsilon U_2(s).
\tag{15}\label{eq:balanced-interim-utility-embedding}
\]
When the agents' true types are \(t_r\) and \(t_s\) and both report truthfully, \eqref{eq:balanced-gross-utility-embedding} also yields
\[\sum_{k\in\mathcal K}q_{rs}^k(u(t_r,\lambda_k^+)+u(t_s,\lambda_k^-)) = 
a_r+a_s
+\varepsilon
\sum_{k\in\mathcal K}q_{rs}^k
(H_{rk}-H_{sk}).
\]
Writing \(\tilde W(q)\) for expected gross social surplus in \(\mathcal E_\lambda\), averaging over the independent uniform types yields
\[
\tilde W(q)
=
\frac{2}{3}\sum_{r\in\mathcal T}a_r
+\varepsilon W(q).
\tag{16}\label{eq:balanced-surplus-embedding}
\]
\eqref{eq:balanced-interim-utility-embedding} and \eqref{eq:balanced-surplus-embedding} imply that the bijection preserves payoff equivalence.
\end{proof}

\begin{corollary}\label{cor:local-utility-triangle-obstruction}
Suppose that three types and three prizes generate a utility triangle with nonzero determinant. There is a two-agent environment with those three types as the common type set and three lottery-valued alternatives. Under independent uniform priors, this environment has a deterministic BIC mechanism with no payoff-equivalent DIC mechanism, even when DIC allocation rules are allowed to randomize. The lotteries can be chosen so that \(\lambda_k^++\lambda_k^-=2\lambda^\circ\) as finite measures on \(X\) for every \(k\in\mathcal K\).
\end{corollary}
\begin{proof}[Proof of \Cref{cor:local-utility-triangle-obstruction}]
Apply \Cref{lem:balanced-embedding} to the deterministic BIC mechanism in \Cref{lem:balanced-obstruction}. The corresponding mechanism in \(\mathcal E_\lambda\) is BIC and retains the same deterministic allocation rule. If this mechanism admitted a payoff-equivalent DIC mechanism, \Cref{lem:balanced-embedding} would imply that the deterministic BIC mechanism in \(\mathcal E_H\) also admitted a payoff-equivalent DIC mechanism, contradicting \Cref{lem:balanced-obstruction}.
\end{proof}

\subsection{Restoring the full type space}

By \Cref{cor:local-utility-triangle-obstruction}, there is a two-agent environment with the three selected types as the common type set and three lottery-valued alternatives. Under independent uniform priors, this environment admits a deterministic BIC mechanism with no payoff-equivalent DIC mechanism. This is almost, but not quite, enough for \Cref{thm:maximal-domain}: I define universal robustifiability using the original type set \(T\) and independent full-support priors. In the next and final lemma, I close precisely this gap. There, I show that any obstruction on a core set of types can be extended to the full type space while assigning positive probability to every type.



Consider a two-agent environment with set of types \(T\), fixed finite alternatives, and fixed lotteries \(\lambda_i(a)\) for each agent and alternative. For each agent \(i\in\{1,2\}\), let \(T_i^\circ\subseteq T\) be a nonempty core set of types, and let \(T_{-i}^\circ\) denote the other agent's core set of types. The core environment is the restriction of the environment to \(T_1^\circ\times T_2^\circ\), under independent uniform priors over the two core sets. I call the reports \(r_i\in T_i^\circ\) agent \(i\)'s \textit{core reports}.

\begin{lemma}\label{lem:full-support-extension}
Suppose that the core environment admits a BIC mechanism \(M^\circ\) with no payoff-equivalent DIC mechanism. Then, under some independent full-support prior on \(T^2\), the full environment admits a BIC mechanism with no payoff-equivalent DIC mechanism.
\end{lemma}

\begin{proof}[Proof of \Cref{lem:full-support-extension}]
Write \(M^\circ=(q^\circ,p^\circ)\). For \(i\in\{1,2\}\), \(t_i\in T\), and \(r_i\in T_i^\circ\), let
\[
V_i(t_i,r_i)
\coloneqq
\frac{1}{|T_{-i}^\circ|}
\sum_{r_{-i}\in T_{-i}^\circ}
\left(
\sum_{a\in\mathcal A}
q^{\circ,a}(r_i,r_{-i})u(t_i,\lambda_i(a))
-p_i^\circ(r_i,r_{-i})
\right).
\]
Choose a map \(\phi_i\colon T\to T_i^\circ\) such that \(\phi_i(t_i)\in\arg\max_{r_i\in T_i^\circ}V_i(t_i,r_i)\) for every \(t_i\in T\). For each \(r_i\in T_i^\circ\), choose \(\phi_i(r_i)=r_i\), which is possible because \(M^\circ\) is BIC in the core environment.

For a probability distribution \(\rho\) on \(T\), write
\((\phi_i)_\#\rho(r)
\coloneqq
\sum_{\substack{t\in T\\\phi_i(t)=r}}\rho(t)\) for the distribution on \(T_i^\circ\) induced by \(\phi_i\).

Let \(m_i\coloneqq|T_i^\circ|\). If \(T\setminus T_i^\circ=\emptyset\), let \(\pi_i^\eta\) be uniform on \(T_i^\circ=T\). If \(T\setminus T_i^\circ\neq\emptyset\), choose positive weights \(w_i(t)\), indexed by \(t\in T\setminus T_i^\circ\), such that
\(\sum_{t\in T\setminus T_i^\circ}w_i(t)=1\). For every sufficiently small \(\eta>0\), set \(\pi_i^\eta(t)\coloneqq\eta w_i(t)\) for \(t\in T\setminus T_i^\circ\) and
\[
\pi_i^\eta(r)
\coloneqq
\frac{1}{m_i}
-
\eta
\sum_{\substack{t\in T\setminus T_i^\circ\\\phi_i(t)=r}}
w_i(t),
\]
for \(r\in T_i^\circ\). For sufficiently small \(\eta\), \(\pi_i^\eta(r)>0\) for every \(r\in T_i^\circ\). In either case, \(\pi_i^\eta\) has full support and satisfies
\[
(\phi_i)_\#\pi_i^\eta
=
\operatorname{Unif}(T_i^\circ),
\tag{17}\label{eq:full-support-pushforward}
\]
and \(\pi_i^\eta(T\setminus T_i^\circ)
\rightarrow 0\) as \(\eta\downarrow0\). Write \(\pi^\eta\coloneqq\pi_1^\eta\otimes\pi_2^\eta\) for the associated product prior.

Define a mechanism \(M^\eta=(q^\eta,p^\eta)\) in the full environment by
\[q^{\eta,a}(t_1,t_2) \coloneqq
q^{\circ,a}(\phi_1(t_1),\phi_2(t_2)) \qquad \text{and} \qquad
p_i^\eta(t_1,t_2) \coloneqq
p_i^\circ(\phi_1(t_1),\phi_2(t_2)).
\]
By \eqref{eq:full-support-pushforward}, the distribution of the opponent's induced core report is uniform. Truthfully reporting \(t_i\) induces the core report \(\phi_i(t_i)\), which maximizes \(V_i(t_i,\cdot)\). Every deviation in the full report space induces another core report. Hence, \(M^\eta\) is BIC under \(\pi^\eta\).

Suppose for the sake of contradiction that every sufficiently small \(\eta>0\) admits a payoff-equivalent DIC mechanism \((\hat q^\eta,\hat p^\eta)\). Fix baseline core reports \(r_i^0\in T_i^\circ\), and normalize the payments by setting
\(\overline p_i^\eta(t_i,t_{-i})
\coloneqq
\hat p_i^\eta(t_i,t_{-i})
-
\hat p_i^\eta(r_i^0,t_{-i})\). Subtracting a term that depends only on the other agent's report preserves DIC.

Because the sets of types and alternatives are finite, choose \(B<+\infty\) such that \(|u(t_i,\lambda_i(a))|\leq B\) for every \(i\), \(t_i\), and \(a\). The DIC inequality for true type \(t_i\) against report \(r_i^0\) implies \(\overline p_i^\eta(t_i,t_{-i})\leq2B\), while the DIC inequality for true type \(r_i^0\) against report \(t_i\) implies \(\overline p_i^\eta(t_i,t_{-i})\geq-2B\). Thus, all normalized payments are uniformly bounded.

Choose a sequence \(\eta_n\downarrow0\). The allocation rules lie in the compact set \(\Delta(\mathcal A)^{T^2}\), and the normalized payment vectors are uniformly bounded in a finite-dimensional space. Hence, after passing to a subsequence, \(\hat q^{\eta_n}\to q^0\) and
\(\overline p^{\eta_n}\to\overline p^0\) for some
\((q^0,\overline p^0)\). Since DIC is defined by finitely many weak inequalities, \((q^0,\overline p^0)\) is DIC and so its restriction to the core report profiles is DIC.

Let \(\overline U_i^n(r_i)\) be the interim utility of core type \(r_i\) under \((\hat q^{\eta_n},\overline p^{\eta_n})\) and prior \(\pi_{-i}^{\eta_n}\). Because the types are independent, payment normalization shifts all of agent \(i\)'s interim utilities by the same constant. Since \(M^{\eta_n}\) restricts to \(M^\circ\) on core reports and \eqref{eq:full-support-pushforward} makes the induced opponent report uniform, payoff equivalence consequently yields
\[
\overline U_i^n(r_i)-\overline U_i^n(r_i^0)
=
U_i^{M^\circ}(r_i)-U_i^{M^\circ}(r_i^0), \quad \text{for every } r_i\in T_i^\circ.
\tag{18}\label{eq:full-support-utility-differences}
\]

By the definition of \(\pi_{-i}^{\eta_n}\), the opponent's prior converges to the uniform distribution on \(T_{-i}^\circ\). Let \(\overline U_i^0(r_i)\) denote interim utility under the restriction of \((q^0,\overline p^0)\) to the core type profiles and the uniform prior on \(T_{-i}^\circ\). Because the sums defining interim utility are finite, the convergence of \(\hat q^{\eta_n}\), \(\overline p^{\eta_n}\), and \(\pi_{-i}^{\eta_n}\) implies \(\overline U_i^n(r_i)\to\overline U_i^0(r_i)\) for every \(r_i\in T_i^\circ\). Taking limits in \eqref{eq:full-support-utility-differences} delivers \(\overline U_i^0(r_i)-\overline U_i^0(r_i^0)=U_i^{M^\circ}(r_i)-U_i^{M^\circ}(r_i^0)\). Set \(\kappa_i\coloneqq\overline U_i^0(r_i^0)-U_i^{M^\circ}(r_i^0)\), and set \(p_i^0(r_i,r_{-i})\coloneqq\overline p_i^0(r_i,r_{-i})+\kappa_i\). Adding \(\kappa_i\) to every payment preserves DIC and reduces every interim utility by \(\kappa_i\). For every \(r_i\in T_i^\circ\), \(\overline U_i^0(r_i)-\kappa_i=\overline U_i^0(r_i)-\overline U_i^0(r_i^0)+U_i^{M^\circ}(r_i^0)=U_i^{M^\circ}(r_i)\). Thus, the restriction of \((q^0,p^0)\) to the core type profiles gives every core type the same interim utility as \(M^\circ\).

It remains to compare social surplus. Since the relevant sums are finite, \(\hat q^{\eta_n}\to q^0\) and \(\pi^{\eta_n}\to\operatorname{Unif}(T_1^\circ)\otimes\operatorname{Unif}(T_2^\circ)\) imply \(W^{\pi^{\eta_n}}(\hat q^{\eta_n})
\rightarrow 
W^{\operatorname{Unif}(T_1^\circ)\otimes\operatorname{Unif}(T_2^\circ)}(q^0)\). Likewise, because \(q^{\eta_n}\) agrees with \(q^\circ\) on the core and the probability of every profile involving an extra type tends to zero,
\(W^{\pi^{\eta_n}}(q^{\eta_n})
\rightarrow 
W^{\operatorname{Unif}(T_1^\circ)\otimes\operatorname{Unif}(T_2^\circ)}(q^\circ)\). By payoff equivalence, \(W^{\pi^{\eta_n}}(\hat q^{\eta_n})=W^{\pi^{\eta_n}}(q^{\eta_n})\) for every \(n\) and so the two limits are equal. Consequently, after adding the constants \(\kappa_i\) to the payments, the restriction of \((q^0,p^0)\) to the core is DIC and payoff-equivalent to
\(M^\circ\), contradicting the hypothesis.
\end{proof}

\section{Proof of \texorpdfstring{\Cref{thm:maximal-domain}}{theorem}}

\begin{proof}[Proof of \Cref{thm:maximal-domain}]
The equivalence among the rank condition, the factorization in \eqref{eq:rank-one-utility-factorization}, and the vanishing of every utility triangle follows from \Cref{lem:incentive-rank-characterization}.

If \(r(u)\leq1\), \Cref{prop:rank-one-robustification} implies that the domain is universally robustifiable. Suppose instead that \(r(u)\geq2\). By \Cref{lem:incentive-rank-characterization}, there are types \(t_0,t_1,t_2\in T\) and prizes \(x_0,x_1,x_2\in X\) whose utility triangle has nonzero determinant. By \Cref{cor:local-utility-triangle-obstruction}, the restriction to those three types admits a two-agent, three-alternative BIC mechanism with no payoff-equivalent DIC mechanism. Its allocation rule is deterministic, and, for every alternative \(k\), the two agents' lotteries satisfy \(\lambda_k^++\lambda_k^-=2\lambda^\circ\).

Let the full environment have type set \(T\) and the same three alternatives, with alternative \(k\) giving the two agents \(\lambda_k^+\) and \(\lambda_k^-\). For each agent \(i\in\{1,2\}\), let \(T_i^\circ\coloneqq\{t_0,t_1,t_2\}\). Applying \Cref{lem:full-support-extension} produces an independent full-support prior over the original set of types \(T\) and a BIC extension with no payoff-equivalent DIC mechanism, which uses the same three alternatives and retains the deterministic allocation rule and the balancing identities. Thus, the domain is not universally robustifiable.\end{proof}

\section{Relation to the literature}\label{sec:related-literature}

The modern BIC-DIC equivalence result begins with \citet{manelli2010bayesian} and is developed for general social-choice environments by \citet{gershkov2013equivalence}.\footnote{The robustness appeal of dominant strategies is emphasized by \citet{bergemann2005robust}. Earlier, \citet{mookherjee1992dominant} asks when a BIC allocation rule itself can be implemented in dominant strategies by changing transfers. Subsequent work studies correlated types, nonlinear utilities, or pursues geometric approaches; see \citet{kushnir2015sufficiency}, \citet{kushnir2019equivalence}, \citet{goeree2023geometric}, \citet{kushnir2020linear}, and \citet{kleiner2021extreme}. \citet{chen2019equivalence} equates stochastic and deterministic mechanisms.} I prove the inverse theorem. I begin with an arbitrary finite utility domain, impose no order or parameterization on its types, and ask when equivalence survives every finite social-choice environment and every independent full-support prior. \Cref{thm:maximal-domain} delivers an exact answer: universal robustifiability holds if and only if incentive rank is at most one. In this precise finite-domain sense, my result completes the positive equivalence program. The scalar-linear structure behind \citet{gershkov2013equivalence}'s theorem is not merely sufficient or convenient; universal equivalence forces it.

The necessity direction of my result is in some sense canonical: any second independent pattern in how types value prize changes contains a three-type, three-prize obstruction. I turn that local failure of rank one into a two-agent environment with three lottery-valued alternatives and an independent full-support prior in which a deterministic BIC mechanism cannot be matched by any payoff-equivalent DIC mechanism, even when DIC allocation rules are allowed to randomize. Moreover, every alternative can preserve the same aggregate lottery across the two agents.\footnote{\citet{gershkov2013equivalence} exhibit failures outside their benchmark assumptions. \citet{manelli2019dominant} study multi-object trading, where fixed supplies create additional feasibility restrictions; and \citet{yang2025multidimensional} establish a sharp anti-equivalence result when the DIC replacement must be deterministic.}

The closest antecedent to my paper is \citet{kushnir2019equivalence}. They begin with ordered one-dimensional types and regular nonlinear valuations, characterize increasing differences over distributions by the one-factor representation \(v_i(a,x_i)=f_i(a)M_i(x_i)+m_i(x_i)+g_i(a)\), and obtain exact BIC-DIC equivalence when \(f=(f_i)_i\) is convex-valued and \(g=(g_i)_i\) is a linear transformation of \(f\).\footnote{They do not establish that the (universal) one-factor structure is necessary for BIC-DIC equivalence. My theorem answers a more basic inverse question. Once the utility domain is fixed but the social-choice environment may vary, what does universal equivalence itself reveal about preferences? A scalar incentive index.} Also relevant is \citet{kartik2024singlecrossing}, whose monotonic-differences specialization has a one-factor form. \citet{kartik2025convex} tie convex choice to directional single-crossing differences and show that, under regularity, directional single-crossing differences over lotteries leave only an effectively one-dimensional or affine representation.\footnote{Related results on aggregation include \citet{quah2012aggregating} and \citet{abbas2011oneswitch}. \citet{carroll2012local} studies when local incentive constraints suffice for global incentive compatibility.} Those papers answer questions about the geometry of choice and comparative statics. The same one-factor form acquires an exact maximal implementation meaning here: it delineates precisely the boundary of universal BIC-DIC equivalence. 

\bibliography{sample}

\end{document}